\documentclass[11pt]{amsart}
\usepackage[latin1]{inputenc}
\usepackage[all]{xy}
\usepackage{amssymb, amsmath, amscd, geometry}
\usepackage{graphicx}
\usepackage{setspace}
\numberwithin{equation}{section}
\input xy
\xyoption{all}
\usepackage{latexsym}
\usepackage[usenames]{color}
\usepackage{epic} 
\usepackage{mathrsfs}
\usepackage{euscript}
\usepackage{rotating}
\usepackage{bbm}

\usepackage{verbatim}
\usepackage[usenames]{color}

\def \rank{{\mathrm{rank}} \, }

\theoremstyle{plain}
\newtheorem{lemma}{Lemma}[section]
\newtheorem{theorem}[lemma]{Theorem}
\newtheorem{corollary}[lemma]{Corollary}

\newtheorem{prop}[lemma]{Proposition}

\theoremstyle{definition}

\newtheorem{remark}{Remark}[section]

\newtheorem*{theorem-non}{Theorem}

      \newcommand{\N}{{\mathbb N}}
      \newcommand{\R}{{\mathbb R}}
      \newcommand{\C}{{\mathbb C}}

\newcommand{\vr}{\varepsilon}

              \newcommand{\Id}{\operatorname{Id}}

\newcommand{\expe}{\mathbb{E}}
\newcommand{\prob}{\mathbb{P}}
\newcommand{\svr}{\sqrt{\vr}}
\newcommand{\trace}{\mathrm{tr}}

\newcommand{\h}{{\mathcal H}}

\newcommand{\eps}{\epsilon}

\newcommand{\uno}{\mathbbm{1}}

\allowdisplaybreaks

\begin{document}

\pagestyle{headings}

\author{M. Junge, T. Oikhberg  and C. Palazuelos}

\title{Reducing the number of inputs in nonlocal games}

\thanks{The first author was partially supported by the NSF DMS-1201886. The second author was partially supported by Simons Foundation travel award 210060. The third author was partially supported by MINECO (grant MTM2011-26912), the european CHIST-ERA project CQC (funded partially by MINECO grant PRI-PIMCHI-2011-1071) and ``Ram\'on y Cajal'' program. The first and third authors are partially supported by ICMAT Severo Ochoa Grant SEV-2011-0087 (Spain).}

\begin{abstract}
In this work we show how a vector-valued version of Schechtman's empirical method can be used to reduce the number of inputs in a nonlocal game $G$ while preserving the quotient $\beta^*(G)/\beta(G)$ of the quantum over the classical bias. We apply our method to the Khot-Vishnoi game, with exponentially many questions per player, to produce another game with polynomially many ($N\approx n^8$) questions so that the quantum over the classical bias is $\Omega (n/\log^2 n)$. 
\end{abstract}

\maketitle

\section*{Introduction and main result}
A remarkable feature of quantum mechanics is the fact that two observers, each holding half of an entangled quantum state, can perform suitable measurements to produce some probability distributions which cannot be explained by a Local Hidden Variable Model. This was first showed by Bell \cite{Bell}, based on a previous intuition of Einstein, Podolski and Rosen in \cite{EPR}, and the nowadays routine experimental verification of this phenomenon \cite{AGR81, hensen2015experimental} provides the strongest evidence that Nature does not obey the laws of classical mechanics.

A Bell experiment can be understood by means of the so called \emph{nonlocal games}. In a bipartite game $G$, Alice and Bob are asked questions $x$ and $y$ respectively according to a fixed and known probability distribution $\pi$, and they are required to answer outputs $a$ and $b$ respectively. Let us assume that $x,y\in \{1,\cdots, N\}$ and $a,b\in \{1,\cdots, K\}$, although the setting could be more general. For each pair of questions $(x, y)\in [N]\times [N]$ and pair of answers $(a, b)\in [K]\times [K]$ there is a known probability $V(x,y,a,b)\in [0,1]$ of winning the game, so that the game $G$ is completely determined by $\pi$ and $V$. Then, the aim of the players is to maximize the bias of the game\footnote{The situation where the players maximize the value of the game is also very interesting, but the optimization of the bias is more suitable in this work.}, defined as the absolute value of the difference between the winning probability and $1/2$. To this end, the players can agree on a strategy 
before the game starts, which is completely described by the numbers $P(a,b|x,y)$ giving the probability of answering the couple $(a,b)$ if they are asked the couple $(x,y)$, but they are not allowed to communicate to each other once the game has started.

A \emph{classical strategy} $P$  for Alice and Bob is given by some functions $f_A:[N]\rightarrow [K]$, $f_B:[N]\rightarrow [K]$ so that we define $P(a,b|x,y)=1$ if $f_A(x)=a$ and $f_B(y)=b$ and $P(a,b|x,y)=0$ otherwise. Then, the \emph{classical bias} of the game $\beta(G)$ is defined as the largest bias of the game where the players use classical strategies. A \emph{quantum strategy} between Alice and Bob is of the form $P(x,y|a,b)=tr(E_x^a\otimes F_b^y \rho)$ for every $x,y,a,b$, where $\{E_x^a\}_{x,a}$ is a set of nonnegative operators acting on a Hilbert space $\mathcal H_A$ such that $\sum_a E_x^a=\uno$ for every $x$ (and analogously for $\{E_y^b\}_{y,b}$) and $\rho$ is a density operator (nonnegative operator with trace one) acting on $\mathcal H_A\otimes \mathcal  H_B$. The \emph{quantum bias} of the game $\beta^*(G)$ is defined as the largest bias when the players use quantum strategies. 

The existence of quantum probability distributions which cannot be explained by a Local Hidden Variable Model is equivalent to the existence of certain games for which the quantum bias is strictly larger than the classical bias. A famous example is given by the CHSH game \cite{CHSH}, where each player is asked a random bit ($N=2$) and they must reply with one bit each ($K=2$). The players win the game if and only if the XOR of the answers is equal to the AND of the questions. Note that here $V(x,y,a,b)\in \{0,1\}$ for every $x,y,a,b$. It is well known that the classical bias is at most 1/4, while  the quantum bias is $\cos(\pi/8)^2-1/2 \approx  0.35$. 

The aim of this work is to study how much quantum mechanics can deviate from classical mechanics, and a natural way to quantify this deviation is via the quotient $\beta^*(G)/\beta(G)$. This quantity has been deeply studied during the last years and, beyond its theoretical interest, it has been shown to be very useful regarding its applications to different contexts such as dimension witnesses, communication complexity, the study of quantum nonlocality in the presence of noise or/and detector inefficiencies and so on \cite{JPPVW}. In fact, in order to consider all relevant parameters in the problem we will denote by $\beta_d^*(G)$ the quantum bias of the game $G$ when the players are restricted to the use of $d$-dimensional quantum states $\rho$ in the corresponding quantum strategies, so that $\beta^*(G)=\sup_d\beta_d^*(G)$. It was proved in \cite{JPI} (see also \cite{Lo12}, \cite{PV-Survey}) that there is a universal constant $C>0$ such that for every bipartite game $G$ with $N$ questions and $K$ answers 
per player we have 
\begin{align}\label{upper bounds}
\frac{\beta_d^*(G)}{\beta(G)}\leq C  \ell, \text{     }\text{     }\text{   where    }  \text{     }\text{       } \ell=\min\{N,K,d\}.
\end{align}

A prominent example of a game leading to a large quotient $\beta^*(G)/\beta(G)$ was given in
\cite{BRSW}, where the authors showed that the so called Khot-Vishnoi game $G_{KV}$ \cite{KhVi},
defined with $N=2^n/n$ questions and $K=n$ answers per player, leads to a quotient
$\beta_n^*(G_{KV})/\beta(G_{KV})\geq Cn/\log^2 n$, where $C$ is a universal constant
\footnote{In \cite{BRSW} the authors study the quotient of the quantum over 
the classical value of the game rather than the bias.}.
According to (\ref{upper bounds}), this example is essentially optimal in the number of answers $K$ and in the dimension of the Hilbert space $d$. However, the number of questions is exponentially far away from the best known upper bound $O(N)$. Our main theorem states that the same estimate can be obtained with polynomially many inputs.  
\begin{theorem}\label{main theorem}
There exists a game $G$ with $N\approx n^8$ questions and $n$ answers per player such that 
$$\frac{\beta_n^*(G)}{\beta(G)}\geq C\frac{n}{\log^2 n}$$for a certain universal constant $C$. 
\end{theorem}
Although our result is still far from the best upper bound $O(N)$ for $\frac{\beta^*(G)}{\beta(G)}$ in the number of questions, Theorem \ref{main theorem} shows that the use of exponentially many questions is not needed to obtain (almost) optimal estimates in the rest of the parameters, as it could be guessed from the example in \cite{BRSW}. We must note that in \cite{JPI} the authors gave an example of a nonlocal game $G_{JP}$ with
$n$ questions and $n$ answers per player for which
$\beta_n^*(G_{JP})/\beta(G_{JP})\geq \sqrt{n/\log n}$.
This is only quadratically far from the best upper bounds in all the parameters of interest
at the same time. Hence, the key point in Theorem \ref{main theorem} is that it preserves
the optimality of the game $G_{KV}$ shown in \cite{BRSW}. In fact, the $n$-dimensional quantum
state used in \cite{BRSW}, as well as in Theorem \ref{main theorem}, to obtain the corresponding
lower bound for $\beta_n^*(G)$ is the maximally entangled state; and it is known \cite{Palazuelos}
that in this case the corresponding quotient $\beta_n^*(G)/\beta(G)$ is upper bounded by
$C\frac{n}{\sqrt{\log n}}$, for a universal constant $C$. So the logarithmic factor cannot be removed.

As we will explain in Section \ref{sec comments}, our procedure is very general and it can be applied in many other contexts. From a physical point of view, passing from an exponential number of parameters to polynomially many can make an important difference for experimental realizations, since the former situation is unapproachable in practice. On the other hand, the possibility of reducing the number of questions in a game while preserving certain properties can also be of independent interest in theoretical computer science, where two-prover one-round games play a central role.

The main reason to work with the bias of a game $G$ rather than with its value is that in the former case we can adopt a much more general point of view. More precisely, for a family of real coefficients $(M_{x,y}^{a,b})_{x,y;a,b}$, we define the \emph{Bell functional} $M$ acting on the set of strategies by means of the dual action $\langle M, P\rangle :=\sum_{x,y;a,b}M_{x,y}^{a,b}P(a,b|x,y)$. Then, we define its classical and quantum value, respectively, as 
\begin{align}\label{general Bell functional}
\omega(M)=\sup_{P\in \mathcal P_c}|\langle M, P\rangle|\text{      } \text{      }\text{      and     }\text{      } \text{        } \omega^*(M)=\sup_{P\in \mathcal P_q}|\langle M, P\rangle|,
\end{align}where $\mathcal P_c$ and $\mathcal P_q$ denote respectively the set of classical and quantum strategies defined above. It is not difficult to see that there is a one-to-one correspondence between the quotient $\beta_n^*(G)/\beta(G)$ between the quantum and the classical bias of games and the quotient $\omega_n^*(M)/\omega(M)$ between the quantum and the classical value of Bell functionals \cite[Section 2]{BRSW}. In particular, if the element $M$ is a game itself so that $M_{x,y}^{a,b}=\pi(x,y)V(x,y,a,b)$ for very $x,y,a,b$, the values in (\ref{general Bell functional}) coincide with its classical and quantum value respectively. The Khot-Vishnoi game $G_{KV}$ used in \cite{BRSW} actually gives a large quotient $\omega_n^*(G_{KV})/\omega(G_{KV})\geq Cn/\log^2 n$. However, as we just said above, one can easily find another game, which we will also denote by $G_{KV}$, for which the same estimates hold for the bias. The possibility of considering real (non necessarily positive) coefficients  will be 
important in our work, as we explain in Section \ref{sec comments}.

The relevance of equation (\ref{general Bell functional}) is that it allows to connect the study of Bell inequalities and nonlocal games to the theory of operate spaces \cite{PV-Survey}.  In particular, if we realize $M=\sum_{x,y,a,b}M_{x,y}^{a,b}(e_x\otimes e_a)\otimes (e_y\otimes e_b)$ as an element in $\ell_1^{N}(\ell_\infty^{K})\otimes \ell_1^{N}(\ell_\infty^{K})$, one can see that
\begin{align}\label{min-epsilon}
\omega(M)\approx\|M\|_{\ell_1^{N}(\ell_\infty^{K})\otimes_\epsilon \ell_1^{N}(\ell_\infty^{K})}  \text{   }\text{  and  }\text{   }
 \omega^*(M)\approx\|M\|_{\ell_1^{N}(\ell_\infty^{K})\otimes_{min} \ell_1^{N}(\ell_\infty^{K})},
\end{align}
where here $\epsilon$ denotes the injective tensor norm when $\ell_1^{N}(\ell_\infty^{K})$ is seen as a complex Banach space and $min$ denotes the minimal tensor norm for which one must understand $\ell_1^{N}(\ell_\infty^{K})$ as an operator space \cite{PiLp}. The symbol $\approx$ can be shown to be an equality if the element $M$ has positive coefficients (in particular, when $M$ is a game). Otherwise, the relation is more subtle. However, the important point for us is that finding an element $M$ for which the quotient $\|M\|_{\ell_1^{N}(\ell_\infty^{K})\otimes_{min} \ell_1^{N}(\ell_\infty^{K})}/\|M\|_{\ell_1^{N}(\ell_\infty^{K})\otimes_\epsilon \ell_1^{N}(\ell_\infty^{K})}$ is large immediately provides an element $\tilde{M}$ for which $\omega^*(\tilde{M})/\omega(\tilde{M})$ is large (see Proposition \ref{connection min-viol} for details). As shown in previous works, connection (\ref{min-epsilon}) is extremely useful since it allows the use of all the machinery of Banach/operator spaces.

The main tool to prove Theorem \ref{main theorem} is a vector-valued version of Schechtman's empirical method \cite{Sch87}. This method allows to reduce the dimension $N$ when one has an embedding $E\hookrightarrow  \ell_1^{N}(X)$ and the space $E$ has low dimension. The reader will quickly understand that the previous description of $\omega(M)$ and $\omega^*(M)$ by means of tensor norms makes the empirical method a natural tool to be applied  in our context. However, in order to prove our result we will need to save some obstacles. The first one is that the $min$ norm behaves well only in the operator space category; hence, we will need to deal with an extra matrix structure involved in its definition. The second main difficulty is that the empirical method applies on low dimensional spaces. However, as we will see below we do not have that property when we deal with the Khot-Vishnoi game. Then, we will need to ``cut" the game in order to apply our method (see Section \ref{s:Main result} for details).

The paper is organized as follows. In Section \ref{s:preliminaries} we will introduce some basic tools about operators spaces and its connections to nonlocal games. In Section \ref{s:empirical} we will prove a vector-valued version of Schechtman's empirical method, needed to prove our main theorem, and we will postpone the analysis of some possible improvements of our results to Section \ref{sec: improving empirical method}.  In Section \ref{s:Main result} we will first analyze the Khot-Vishnoi game from a mathematical point of view to highlight some of its properties. Then, we will see how the empirical method can be applied on it to obtain Theorem \ref{main theorem}.
\section{Preliminaries on operator spaces}\label{s:preliminaries}
We start by recalling some basic ideas from operator space theory and its connection to nonlocal games; further details can be found in \cite{EffrosRuan}, \cite{PiOS} and \cite{PV-Survey}.
An \emph{operator space} $X$ is a closed subspace of $\mathcal B(\h)$, the bounded operators on a Hilbert space $\h$. The inclusions $M_d(X) \subseteq M_d(\mathcal B(\h))\simeq \mathcal B(\h^{\otimes d})$ induce \emph{matrix norms} $\|\cdot \|_d$ on $M_d(X)$. Ruan's Theorem (\cite[Theorem 2.3.5]{EffrosRuan}) states that such matrix norms also characterize operator spaces. More precisely, the existence of such an inclusion into $\mathcal B(\h)$ is equivalent to having a sequence of matrix norms $(M_d(X),\|\cdot\|_d)$ satisfying the following conditions: for every integers $c$ and $d$,
\begin{itemize}
\item $\|v \oplus w\|_{c+d} = \max\{\|v\|_c,\|w\|_d \}$ for every $v\in M_c(X)$ and $w\in M_d(X)$, 
\item $\|\alpha \cdot v \cdot \beta\|_c \leq \|\alpha\| \, \|v\|_c \, \|\beta\|$ for every $\alpha, \beta \in M_c$ and $v\in M_c(X)$. 
\end{itemize}

To specify an operator space structure, one can either provide an explicit inclusion of $X$ into $\mathcal B(\h)$ or describe the matrix norms on $M_d(X)$ for every $d\geq 1$. A canonical example of an operator space is $M_N$, the space of complex matrices of order $N$, with its operator space structure given by the usual identification $M_N \simeq \mathcal B(\mathbb C^N)$. The matrix norm $\|\cdot\|_d$ on $M_d(M_N)$ is then the usual operator norm on $M_{dN}$. This identification also induces a natural operator space structure on $\ell_\infty^N$, by identifying this space with the diagonal of $M_N$. In particular, given an element $\sum_{i=1}^Nx_i\otimes e_i\in M_d\otimes \ell_\infty^N$, the corresponding norm is given by $$\Big\|\sum_{i=1}^Nx_i\otimes e_i\Big\|_{M_d(\ell_\infty^N)}=\sup_{i=1,\cdots, N}\|x_i\|_{M_N}.$$

Given a linear map $T:X\rightarrow Y$  between operator spaces, let $T_d:=\Id_{M_d} \otimes T: M_d(X)\rightarrow M_d(Y)$ denote the new linear map defined by 
$$T_d((v_{ij})_{i,j}):=(T(v_{ij}))_{i,j}.$$
The map $T$ is said to be \emph{completely bounded} if its completely bounded norm if finite:
$$
\|T\|_{cb} \,:=\, \sup_{d\in\N} \|T_d\|\,<\, \infty.
$$

In fact, we can define analogously the notion of complete contraction, compete isomorphism, complete isometry and so on, by requiring to have the corresponding property when tensorizing with $M_d$. 

As in the Banach space category, one can define the \emph{dual operator space} $X^*$ of the operator space $X$, with matrix norms given by
\begin{align}\label{duality}
M_d(X^*)=CB(X,M_d),
\end{align}where here $CB(X,M_d)$ denotes the space of completely bounded maps from $X$ to $M_d$ endowed with the completely bounded norm. In particular, this allows us to define an operator space structure on $\ell_1^N=(\ell_\infty^N)^*$ and on $S_1^N=M_N^*$.

By way of example, it is well known and easy to check that the operator space dual of
$\ell_\infty^N$ is the space $\ell_1^N$: if $(e_i)_{i=1}^N$ is the canonical basis
of $\ell_1^N$, and $x_i$ are $d \times d$ matrices, then
$$
\Big\| \sum_i x_i \otimes e_i \big\|_{M_d(\ell_1^N)} =
\sup \left\{ \Big\| \sum_i x_i \otimes a_i \big\|_{dm} : 
m \in \N, \max_i \|a_i\|_{M_m} \leq 1 \right\} .
$$

Given two operator spaces $X \subseteq \mathcal B(\h)$ and $Y \subseteq \mathcal B(\mathcal K)$,
their algebraic tensor product $X \otimes Y$ is a subspace of $\mathcal B(\h \otimes \mathcal K)$
and their \emph{minimal} (or \emph{injective}) \emph{operator space tensor product}
$X \otimes_{\min} Y$ is the closure of $X \otimes Y$ in $\mathcal B(\h \otimes \mathcal K)$.
In particular, note that for every operator space $X$, we trivially obtain that
$M_d\otimes_{min} X=M_d(X)$ holds isometrically. If $X$ and $Y$ are finite dimensional,
then we have a natural algebraic identification $X \otimes Y = L(X^*,Y)$,
between the algebraic tensor product and the set of linear maps from $X^*$ to $Y$.
Here, for a given $u =\sum_{i=1}^l x_i \otimes y_i\in X\otimes Y$ one defines
the linear map $T_u:X^*\rightarrow Y$ by $T_u(x^*)=\sum_{i=1}^l x^*(x_i) y_i$
for every $x^*\in X^*$. In addition, given a linear map $T\in L(X^*,Y)$,
we can associate the element $u_T=\sum_{i=1}^sx_i\otimes T(x_i^*)\in X\otimes Y$,
where $(x_i)_{i=1}^s$ and $(x_i^*)_{i=1}^s$ are dual basis of $X$ and $X^*$ respectively.
The previous correspondence induces the following isometric identifications:
\begin{align}\label{isometric ident}
E \otimes_{\eps} F =B(E^*,F), \text{    }\text{    }\text{    }\text{    }\text{    } X \otimes_{\min} Y = CB(X^*,Y)
\end{align}
for $E$ and $F$ finite dimensional Banach spaces and $X$ and $Y$ finite dimensional
operator spaces. We recall that for  $u =\sum_{i=1}^l e_i \otimes f_i\in E\otimes F$,
its $\eps$ norm is defined as
$$\epsilon(u)=\sup\Big\{\Big|\sum_{i=1}^l e^*(e_i) f^*(f_i)\Big|:
 e\in B_{E^*}, \text{   }f\in B_{F^*}\Big\}.$$
Here, $B_Z$ denotes the unit ball of the Banach space $Z$. In analogy, it is not difficult
to see that if $u =\sum_{i=1}^l x_i \otimes y_i\in X\otimes Y$,
$$\|u\|_{\min}=\sup
\Big\{
\Big\|\sum_{i=1}^l T(x_i) \otimes S(y_i)
\Big\|_{\mathcal B(\h\otimes \mathcal K)}
\Big\},$$where the supremum runs now over all complete contractions
$T:X\rightarrow \mathcal B(\h)$, $S:Y\rightarrow \mathcal B(\mathcal K)$.

One can also check that for every operator space $X$, the identification
$\ell_\infty^N(X)=\ell_\infty^N\otimes_{min}X$ holds isometrically.
This identification allows to endow the Banach space $\ell_\infty^N(X)$
with a natural operator space structure, which is the one induced by
$\ell_\infty^N(X)\subseteq \ell_\infty^N(\mathcal B(\h))$.
According to the comments above, one can naturally use duality (\ref{duality})
to obtain an operator space structure on $\ell_1^N(X)$ for every operator space $X$.

In the current work we use the space $\ell_1^N(\ell_\infty^K)$,
with the norm described below. If $(e_i)_{i=1}^N$ and $(f_j)_{j=1}^K$ are
the canonical bases of $\ell_1^N$ and $\ell_\infty^K$ respectively, then,
for $d \times d$ matrices $a_{ij}$, we have
$$
\Big\| \ \sum_{i=1}^N\sum_{j=1}^K a_{ij} \otimes e_i \otimes f_j \Big\|_{M_d(\ell_1^N(\ell_\infty^K))} =
\sup \Big\| \sum_{i=1}^N\sum_{j=1}^K a_{ij} \otimes u_i (f_j) \Big\|_{M_d(B(H))} ,
$$
with the supremum taken over all complete contractions $u_i : \ell_\infty^K \to B(H)$
($u_i$ acts on the $i$-th copy of $\ell_\infty^K$, spanned by
$e_i \otimes f_j$, $1 \leq j \leq K$).
The dual of $\ell_1^N(\ell_\infty^K)$ is the space $\ell_\infty^N(\ell_1^K)$,
whose operator space norm is described, for $d\geq 1$, by
\begin{equation}
\Big\| \sum_{i=1}^N\sum_{j=1}^K a_{ij} \otimes e_i \otimes f_j \Big\|_{M_d(\ell_\infty^N(\ell_1^K))} =
\max_{i=1,\cdots, N} \sup \Big\{ \Big\| \sum_{j=1}^K a_{ij} \otimes u_{ij} \Big\|_{M_{d^2}} :
 \|u_{ij}\|_{M_d} \leq 1 \Big\}.
\label{eq:infty(1)}
\end{equation}
Here, $(e_i)_{i=1}^N$ and $(f_j)_{j=1}^K$ are the canonical bases of
$\ell_\infty^N$ and $\ell_1^K$ respectively.

The following proposition, which was first shown in \cite{JPPVW} and which is explained in much detail in \cite{PV-Survey}, will be crucial in our work.
\begin{prop}\cite[Corollary 4]{JPPVW}\label{connection min-viol}
Let $M=\sum_{x,y,a,b}M_{x,y}^{a,b}(e_x\otimes e_a)\otimes (e_y\otimes e_b)$ be an element in $\ell_1^{N}(\ell_\infty^{K})\otimes \ell_1^{N}(\ell_\infty^{K})$ such that
\begin{align*}
\frac{\|M\|_{\ell_1^{N}(\ell_\infty^{K})\otimes_{min} \ell_1^{N}(\ell_\infty^{K})} }{\|M\|_{\ell_1^{N}(\ell_\infty^{K})\otimes_\epsilon \ell_1^{N}(\ell_\infty^{K})}}\geq \alpha.
\end{align*}
Then, the Bell functional $\tilde{M}=(\tilde{M}_{x,y}^{a,b})_{x,y;a,b}$,  where $x,y=1,\cdots ,N$ and $a,b=1, \cdots, K+1$, obtained by adding extra zeros to the element $M$, verifies 
\begin{align*}
\frac{\omega^*(\tilde{M})}{\omega(\tilde{M})}\geq C\alpha,
\end{align*}where $C$ is a universal constant which can be taken equal to $1/16$.
\end{prop}

In fact, the constant $C$ can be taken $1/4$ if one allows to increase the dimension of the corresponding Hilbert space to compute  $\omega^*(\tilde{M})$, as explained in \cite[Lemma 4.3]{PV-Survey}.

Our construction (given in Section \ref{s:Main result}) yields
a matrix $M$ some of whose coefficients may be negative.
To return to the setting of games (corresponding to matrices
with positive coefficients), we observe that for a given Bell functional $M$ we can define the element $G$ as $$
G_{x,y}^{a,b} = \frac{1}{2N^2} +  \frac{1}{2N^2L}  \tilde{M}_{x,y}^{a,b}, \text{     }\text{  for every    } x,y,a,b;
$$where $L=\max_{x,y,a,b}|M_{x,y}^{a,b}|$ and $N$ is the number of inputs (questions). It is very easy to check that $G$ has positive coefficients with values in $[0,1]$ and it verifies that $$\displaystyle \frac{\beta^*(G)}{\beta(G)}=\frac{\omega^*(M)}{\omega(M)},$$where $\beta(G)$ (resp. $\beta^*(G)$) is the classical (resp. quantum) bias of the game $G$ defined as $$\beta(G)=\sup\Big\{\Big|\langle G,P\rangle-\frac{1}{2}\Big|:P\in \mathcal P_c\Big\}\text{     }\text{  and    }\text{   } \beta^*(G)=\sup\Big\{\Big|\langle G,P\rangle-\frac{1}{2}\Big|:P\in \mathcal P_q\Big\}.$$

This observation leads to the following consequence.
\begin{prop}\label{p:Regev}
Let $M=\sum_{x,y,a,b}M_{x,y}^{a,b}(e_x\otimes e_a)\otimes (e_y\otimes e_b)$ be an element
in $\ell_1^{N}(\ell_\infty^{K})\otimes \ell_1^{N}(\ell_\infty^{K})$ such that
\begin{align*}
\frac{\|M\|_{\ell_1^{N}(\ell_\infty^{K})\otimes_{min} \ell_1^{N}(\ell_\infty^{K})} }{\|M\|_{\ell_1^{N}(\ell_\infty^{K})\otimes_\epsilon \ell_1^{N}(\ell_\infty^{K})}}\geq \alpha.
\end{align*}
Then, there exists a game $G$ with $N$ inputs and $K+1$ outputs, with $\displaystyle \frac{\beta^*(G)}{\beta(G)}\geq C\alpha$, where, as before, we can take $C$ to be $1/16$.
\end{prop}

In this paper we use the notion of the \emph{conjugate space}.
If $X$ is an operator space, then $\overline{X}$
is the same space, but with conjugate multiplication. More specifically,
denote by $\overline{x}$ the element of $x \in X$ when considered as
sitting in $\overline{X}$. Then $\overline{\lambda} \cdot \overline{x}
= \overline{\lambda x}$. Thus, the map $X \to \overline{X} : x \mapsto \overline{x}$
is an antilinear isometry. If $X \subseteq \mathcal B(\h)$, its
\emph{conjugate} operator space structure is given by the embedding
$\overline{X} \subseteq \overline{\mathcal B(\h)} = \mathcal B(\overline \h)$.
We can therefore describe the operator space structure on $\overline{X}$ via
$$
\left\| \sum_i a_i \otimes \overline{x_i} \right\|_{M_d (\overline{X})} =
\left\| \sum_i \overline{a_i} \otimes x_i \right\|_{M_d (X)}
  \text{      }\text{   for every $d$.}
$$
If $X$ consists of matrices $(a_{ij})_{i,j=1}^\infty \in B(\ell_2)$
(with respect to a certain basis of $\ell_2$),
then we can view elements of $\overline{X}$ as matrices $(\overline{a_{ij}})$. We refer the reader to \cite[Section 2.9]{PiOS} for more information.

In general, the formal identity $X \to \overline{X}$ need not be a complete isometry.
However, for the case we are interested in, $X= \ell_\infty^N(\ell_1^K)$, the formal identification of bases yields a linear complete isometry.
Indeed, let $(e_i)_{i=1}^N$ and $(f_j)_{j=1}^K$ be the bases of
$\ell_\infty^N$ and $\ell_1^K$ respectively, then $e_i \otimes f_j$
is the ``canonical'' basis of $X$.
For $d \times d$ matrices $a_{ij}$, \eqref{eq:infty(1)} gives
\begin{align*}
&\Big\| \sum_{i=1}^N\sum_{i=1}^K a_{ij} \otimes e_i \otimes f_j \Big\|_{M_d(X)}=
\max_{i=1,\cdots, N} \sup \Big\{ \Big\| \sum_{j=1}^K a_{ij} \otimes \overline{u_{ij}} \Big\|_{M_{d^2}} :
  \|u_{ij}\|_{M_N} \leq 1 \Big\}\\&
 =\max_{i=1,\cdots, N} \sup \Big\{ \Big\| \sum_{j=1}^K \overline{a_{ij}} \otimes u_{ij} \Big\|_{M_{d^2}} :
 \|u_{ij}\|_{M_d} \leq 1 \Big\}=\Big\| \sum_{i=1}^N\sum_{i=1}^K \overline{a_{ij}} \otimes e_i \otimes f_j \Big\|_{M_d(X)}\\&=
 \Big\| \sum_{i=1}^N\sum_{i=1}^K a_{ij} \otimes \overline{e_i \otimes f_j} \Big\|_{M_d(\overline{X})}.
 \end{align*}

A similar computation shows that the basis $(e_i \otimes f_j)$
is $1$-completely unconditional; that is,
$$
\Big\| \sum_{i=1}^N\sum_{i=1}^K a_{ij} \otimes \alpha_{ij} e_i \otimes f_j \Big\|_{M_d(X)}=
\Big\| \sum_{i=1}^N\sum_{i=1}^K a_{ij} \otimes e_i \otimes f_j \Big\|_{M_d(X)} 
$$
whenever $a_{ij} \in M_d$ and $|\alpha_{ij}| = 1$ for any $i, j$.

As $(\overline{X})^* = \overline{X^*}$ (with conjugate action
$\langle \overline{x^*}, \overline{x} \rangle = \overline{\langle x^*, x \rangle}$),
the two preceding statements hold for the space $X = \ell_1^N(\ell_\infty^K)$
(the dual of $\ell_\infty^N(\ell_1^K)$) as well.
More precisely, if $a_{ij} \in M_d$, and $|\alpha_{ij}| = 1$ for any $(i,j)$, then
$$
\Big\| \sum_{i=1}^N\sum_{i=1}^K a_{ij} \otimes \alpha_{ij} e_i \otimes f_j \Big\|_{M_d(X)}=
\Big\| \sum_{i=1}^N\sum_{i=1}^K a_{ij} \otimes e_i \otimes f_j \Big\|_{M_d(X)} =
\Big\| \sum_{i=1}^N\sum_{i=1}^K a_{ij} \otimes \overline{e_i \otimes f_j} \Big\|_{M_d(\overline{X})} .
$$

Consequently, suppose $y_k \in \ell_1^N(\ell_\infty^K)$ are of the form
$y_k = \sum_{i,j} \alpha_{ij} e_i \otimes f_j$, where $\alpha_{ij}$
are real numbers. Then, for all $a_k \in M_d$, we have
\begin{equation}
\Big\| \sum_{k=1}^M a_k \otimes y_k \Big\|_{M_d(\ell_1^N(\ell_\infty^K))} =
 \Big\| \sum_{k=1}^M a_k \otimes \overline{y_k} \Big\|_{M_d(\overline{\ell_1^N(\ell_\infty^K)})} .
\label{eq:good norm}
\end{equation}
In a similar fashion, we can show that
\begin{equation}
\Big\| \sum_{k=1}^M a_k \otimes y_k \Big\|_{M_d \otimes_\eps \ell_1^N(\ell_\infty^K)} =
 \Big\| \sum_{k=1}^M a_k \otimes \overline{y_k} \Big\|_{M_d \otimes_\eps \overline{\ell_1^N(\ell_\infty^K)}} .
\label{eq:good norm eps}
\end{equation}

We also use the \emph{operator Hilbert space} $OH_N$, introduced by G.~Pisier in \cite{PiOH}.
On the Banach space level, it is the space $\ell_2^N$, with matrix norms given by
$$
\Big\|\sum_{i=1}^N x_i \otimes e_i\Big\|_{M_d(OH_N)}=
\Big\|\sum_{i=1}^N x_i\otimes \overline{x_i}\Big\|_{M_d(\overline{M_d})}^{1/2},
$$
where $(e_i)_{i=1}^N$ is an orthonormal basis of $\ell_2^N$.
In light of the above discussion on complex conjugation, we can view $\overline{x_i}$
as obtained from $x_i$ by entrywise complex cojugation, and
$\sum_{i=1}^N x_i\otimes \overline{x_i}$ as a $d^2 \times d^2$ matrix.
Hence, if each matrix $x_i$
has real entries (in a certain basis), then
$$
\Big\|\sum_{i=1}^N x_i \otimes e_i\Big\|_{M_d(OH_N)}=
\Big\|\sum_{i=1}^N x_i\otimes x_i\Big\|_{M_{d^2}} .
$$

 One can check that the canonical isometric identification $\ell_2^N\simeq\overline{(\ell_2^N)^*}$
at the Banach space level induces a complete isometry from $OH_N$ to $\overline{OH_N^*}$.
In fact, $OH_N$ is the unique operator space with this property, up to complete isometries.

Let us consider a linear map $v:\ell_2^N\rightarrow X$, where $X$ is a Banach space. Let us also fix an orthonormal basis $(\theta_i)_{i=1}^N$ of $\ell_2^N$. Then, it is very easy to see that 
\begin{equation}
\|v\|^2 =
\sup \Big\{ \sum_{i=1}^N\big| \langle f , v \theta_i \rangle \big|^2 : f \in B_{X^*} \Big\} ,
\label{eq:||v||}
\end{equation}
By the definition of the injective tensor product of Banach spaces,
\begin{equation}
\begin{split}
\Big\|\sum_{i=1}^Nv(\theta_i)\otimes \overline{v(\theta_i)}\Big\|_{X\otimes_{\eps} \overline{X}}
&
=
\sup \Big\{ \Big| \sum_{i=1}^N \langle f , v \theta_i \rangle
 \langle \overline{g} , \overline{v \theta_i} \rangle \Big| : f, g \in B_{X^*} \Big\}
\\
&
=
\sup \Big\{ \Big| \sum_{i=1}^N\langle f , v \theta_i \rangle
 \overline{\langle g , v \theta_i \rangle} \Big| : f, g \in B_{X^*} \Big\} .
\end{split}
\label{eq:||v||2}
\end{equation}
Combining H\"older Inequality with \eqref{eq:||v||}, we obtain:
$$
\Big\|\sum_{i=1}^Nv(\theta_i)\otimes \overline{v(\theta_i)}\Big\|_{X\otimes_{\eps} \overline{X}} \leq
\sup \Big\{ \sum_{i=1}^N \big| \langle f , v \theta_i \rangle \big|^2 : f \in B_{X^*} \Big\}\leq \|v\|^2 .
$$
Moreover, plugging $f = g$, for which the supremum in \eqref{eq:||v||} is attained, into
\eqref{eq:||v||2}, we show $$\displaystyle \|v\|^2 =
 \Big\|\sum_{i=1}^Nv(\theta_i)\otimes \overline{v(\theta_i)}\Big\|_{X\otimes_{\eps} \overline{X}}.$$

The previous paragraph proves the first point of the following proposition. The proof of the second part can be found in \cite[Proposition 7.2]{PiOS}.

\begin{prop}\label{OH prop}
Suppose $(\theta_i)_{i=1}^N$ is an orthonormal basis in $OH_N$, $X$ is an operator space,
and $v:OH_N\rightarrow X$ is a linear map. Then,
\begin{enumerate}
\item 
$\displaystyle \|v\|^2 =
 \Big\|\sum_{i=1}^Nv(\theta_i)\otimes \overline{v(\theta_i)}\Big\|_{X\otimes_{\eps} \overline{X}}$.
\item
$\displaystyle \|v\|_{cb}^2 =
 \Big\|\sum_{i=1}^Nv(\theta_i)\otimes \overline{v(\theta_i)}\Big\|_{X\otimes_{min} \overline{X}}$.
\end{enumerate}
\end{prop}

The following version of the preceding result follows directly from the earlier
discussion on complex conjugation.
\begin{corollary}\label{OH real}
Suppose $(\theta_i)_{i=1}^N$ is an orthonormal basis in $OH_N$,
and $v:OH_N\rightarrow M_d$ is a linear map, so that for any $1 \leq i \leq N$,
$v(\theta_i)$ is a real linear combination of matrix units. Then,
$$\displaystyle \|v\|^2 =
 \Big\|\sum_{i=1}^N v(\theta_i)\otimes v(\theta_i)\Big\|_{M_{d^2}}.$$
\end{corollary}

Here and below, the word ``matrix units (in $M_d$)'' refers to $d \times d$
matrices, in which $1$ entry is $1$, and other entries vanish (with respect to
a certain fixed orthonormal basis of $\ell_2^n$). Clearly the matrix units
form a basis for $M_d$.

\section{Vector-valued empirical method}\label{s:empirical}
In this section we prove a vector-valued version of Schechtman's \emph{empirical method} \cite{Sch87}. 
\begin{prop}\label{p:reduction}
Let $E$ be an $m$-dimensional subspace of $L_r(\mu,X)$, where $X$ is a Banach space, $(\Omega,\mu)$ is a probability space and $1 \leq r < \infty$. For a given $\epsilon \in (0,1/2)$ there exists a constant $C(\vr)$ such that if we consider $n = \lceil C(\vr) m^{1+r} \rceil$, then $\ell_r^n(X)$ contains a subspace $E^\prime$ which is $(1+\vr)$-isomorphic to $E$. We can take $C(\vr) = C_0 \vr^{-2} \log(\vr^{-1})$ for a universal constant $C_0$.
\end{prop}
We will first prove an easy lemma which will make the proof of Proposition \ref{p:reduction} simpler. To this end, let us consider a normalized Auerbach basis $(e_i)_{i=1}^m$ in $E$. Note that we clearly have $\|e_i\|_E^r=\int_\Omega \|e_i(t)\|_X^r \, d\mu(t) = 1$ for every $i=1,\cdots, m$. Let us define the function $\phi(t) = m^{-1} \sum_{i=1}^m \|e_i(t)\|_X^r$ for every $t\in \Omega$. It is then obvious that $\int_0^1 \phi(t) \, dt = 1$, so that $d\nu = \phi \, d\mu$ defines a probability measure on $\Omega$.
\begin{lemma}\label{l:inf norm}
Let us define the linear map $S : E \to L_r(\nu,X)$ by  $S(e)=\phi^{-1/r} e$ for every $e\in E$ with the convention
$\frac{0}{0}= 0$. Then, $S$ is an isometry and, moreover, for any $e \in E$ we have $$\|S e\|_{L_\infty(\nu,X)} \leq
m \|e\|_{L_r(\mu,X)} = m \|S e\|_{L_r(\nu,X)}.$$
\end{lemma}
\begin{proof}
Note that, if $\phi(t) = 0$ for some $t$, then $e(t) = 0$ for every $e \in E$. Thus, $S$ is well defined. The linearity of $S$ is obvious. On the other hand, for a given $e\in E$ we have 
\begin{align*}
\|Se\|_{L_r(\nu,X)}^r=\int_\Omega \|\phi^{-1/r}(t) e(t)\|_X^r \phi(t) \, d\mu(t)=\|e\|_{L_r(\mu,X)}^r. 
\end{align*}Hence, $S$ is indeed an isometry. Finally, given a norm one element $e=\sum_{i=1}^m\alpha_i e_i \in E$, the fact that $(e_i)_{i=1}^m$ is an Auerbach basis implies that $\max_i |\alpha_i| \leq 1$. Hence,
\begin{align*}
\|(S e)(t)\|_X = \phi^{-1/r}(t) \|e(t)\|_X =
m^{1/r} \Big( \sum_{i=1}^m \|e_i(t)\|_X^r \Big)^{-1/r}
\Big\| \sum_{i=1}^m \alpha_i e_i(t) \Big\|_X .
\end{align*}
Let $r^\prime = \frac{r}{r-1}$ so that $\frac{1}{r}+ \frac{1}{r'} = 1$. By H\"older Inequality,
\begin{align*}
\Big\| \sum_{i=1}^m \alpha_i e_i(t) \Big\|_X\leq
\Big( \sum_{i=1}^m|\alpha_i|^{r^\prime} \Big)^{1/r^\prime}
\Big( \sum_{i=1}^m \|e_i(t)\|_X^r \Big)^{1/r}
\leq m^{1/r^\prime} \Big( \sum_{i=1}^m \|e_i(t)\|_X^r \Big)^{1/r}.
\end{align*}Hence, for almost every $t$,
\begin{align*}
\|(S e)(t)\|_X \leq m^{1/r} \Big( \sum_{i=1}^m \|e_i(t)\|_X^r \Big)^{-1/r}m^{1/r^\prime} \Big( \sum_{i=1}^m \|e_i(t)\|_X^r \Big)^{1/r}= m .
\end{align*}
\end{proof}
The following lemma is a standard large deviation inequality for sums of independent random variables. The proof can be found in \cite[Lemma 3]{Sch87}.
\begin{lemma}\label{deviation inequality}
Let $(y_i)_{i=1}^n$ be a family of independent random variables and let $A$ and $B$ be non-negative constants such that $\mathbb E y_i=0$, $\mathbb E |y_i|\leq A$ and $\|y_i\|_\infty\leq B$ for every $i=1,\cdots, n$. Then,
\begin{align*}
\mathbb P\Big(\Big|\sum_{i=1}^ny_i\Big|> c\Big)\leq 2\exp \big( -\frac{c^2}{4eABn} \big)
\end{align*}
for all $c\leq 2eAn$.
\end{lemma}
For the proof of Proposition \ref{p:reduction} we use some ideas from
\cite[Section 2]{BLM} and from \cite{Sch87}.
\begin{proof}[Proof of Proposition \ref{p:reduction}]
According to Lemma \ref{l:inf norm} we can assume that $\|e\|_{L_\infty(\mu, X)} \leq m \|e\|_{L_r(\mu, X)}$, for any $e \in E$. On the other hand, for every $t = (t_1, \ldots, t_n) \in \Omega^n$, we consider the linear map $T_t : L_r(\mu, X) \to L_r^n(X)$ defined by  $T_t(f) = (f(t_i))_{i=1}^n$.
Here, $L_r^n$ denotes the space $\mathbb R^n$ endowed with the norm $\|(\alpha_i)_{i=1}^n\|_{L_r^n}^r = n^{-1} \sum_{i=1}^n |\alpha_i|^r$.

Now, let $e \in E$ be any fixed element such that $\|e\|_E= 1$. For $1 \leq i \leq n$ we consider the random variable $y_i:\Omega^n\rightarrow \mathbb R$ defined by $y_i(t) = \|e(t_i)\|_X^r -1$ for every $t$. Then, $(y_i)_{i=1}^n$ is a family of independent random variables satisfying $\expe y_i = 0$, $\expe |y_i| \leq 2$ and $\|y_i\|_\infty \leq m^r$ for every $1 \leq i \leq n$. On the other hand, $\|T_t(e)\|_{L_r^n(X)}^r-1=n^{-1}\sum_{i=1}^ny_i(t)$ for every $t\in \Omega^n$. According to Lemma \ref{deviation inequality} it follows that 
\begin{align*}
\prob \Big(t\in \Omega^n: \big| \|T_t(e)\|_r^r - 1 \big| \geq c \Big) =
\prob \Big(t\in \Omega^n: \big| \sum_{i=1}^ny_i(t) \big| \geq c n \Big) \leq
2 \exp \big( - \frac{c^2 n}{8em^r} \big)
\end{align*}
for any $c \in (0,1)$. Note that, if $\big|\|T_t(e)\|_{L_r^n(X)}^r - 1 \big| \leq c$,
then $1 - c \leq \|T_t(e)\|_{L_r^n(X)} \leq 1 + c$.

Let $\eta = \vr/4$. Standard techniques allow us to find an $\eta$-net $\mathcal{N}$ in the unit sphere of $E$, with $|{\mathcal{N}}| \leq (\frac{3}{\eta})^m=(\frac{12}{\epsilon})^m$. If we consider the particular choice $c=\eta$, we have
\begin{align*}
2 \exp \big( - \frac{\eta^2 n}{8em^r} \big)|{\mathcal{N}}| \leq 2 \exp \big(m\log \frac{12}{\epsilon} - \frac{\eta^2 n}{8em^r} \big)<1,
\end{align*}where the last inequality follows from our choice $n = \lceil C(\vr) m^{1+r} \rceil$ with $C(\vr) = C_0 \vr^{-2} \log(\vr^{-1})$ for a certain universal constant $C_0$. This means that there is an strictly positive probability of having an element $t\in \Omega^n$ such that $1 - \eta \leq \|T_t(e)\|_{L_r^n(X)} \leq 1 + \eta$ for every $e\in \mathcal{N}$. We claim that  
\begin{align}\label{claim}
1 - \vr \leq \|T_t(e)\|_{L_r^n(X)} \leq 1 + \vr 
\end{align}for every $e$ in the unit sphere of $E$. If this is so, we conclude the proof by noting that the map $k:L_r^n(X)\rightarrow \ell_r^n(X)$ defined by $k\big((x_i)_{i=1}^n\big)=n^{-\frac{1}{r}}(x_i)_{i=1}^n$ is an isometry.

To prove claim (\ref{claim}) let us consider an arbitrary unit element $e$ and write it as $e=e_0+\sum_{n=1}^\infty a_n e_n$, with $e_n\in \mathcal{N}$ and $|a_n|\leq \eta^n$ for every $n\geq 0$. Then, we have that 
\begin{align*}
\Big|\|T_t(e)\|_{L_r^n(X)}-\|T_t(e_0)\|_{L_r^n(X)}\Big|\leq \|T_t(e-e_0)\|_{L_r^n(X)}\leq 
\sum_{n=1}^\infty\eta^n\|T_t(e_n)\|_{L_r^n(X)}\leq \frac{\eta}{1-\eta}(1 + \eta).
\end{align*}Then, one has $\frac{1-3\eta}{1-\eta}\leq \|T_t(e_n)\|_{L_r^n(X)}\leq \frac{1+\eta}{1-\eta}$ for every unit element $e$; from where one can deduce our claim by plugging $\eta = \vr/4$.
\end{proof}
\begin{remark}\label{remark: map J}
In this work we will be mostly interested in the space $\ell_1^N(X)$. Let us assume that we have an $m$-dimensional subspace $E\subset \ell_1^N(X)$. Proposition \ref{p:reduction} tells us that if we consider $n = \lceil C(\vr) m^2 \rceil$, then there exists a map $J:\ell_1^N\rightarrow \ell_1^n$ defined by some indices $i_1,\cdots, i_n\in \{1,\cdots, N\}$ and some positive numbers $\alpha_1,\cdots, \alpha_n$ such that $J(x_1,\cdots, x_N)= (\alpha_1x_{i_1},\cdots, \alpha_nx_{i_n})$ for every $(x_1,\cdots, x_N)\in \ell_1^N$ and such that $J\otimes id_X$ defines a $(1+\epsilon)$-isomorphism from $E$ to $\ell_1^n(X)$.
\end{remark}
In Section \ref{sec: improving empirical method} we will explain how
to improve Proposition \ref{p:reduction} if the subspace $E$ of $L_1(\mu, X)$
has some \emph{geometrical properties}. Since we will not use these results in our paper,
we have preferred to postpone this discussion to the end of the paper.
\section{Main result}\label{s:Main result}

In this section we will prove our main theorem. We will start re-proving the classical and
the quantum bounds for the Khot-Vishnoi game in the language of operator spaces.
In particular, we show that this game can be understood as a map factorizing through
a Hilbert space. This fact will be crucial in our analysis later.

Our results deal with operators from $OH$ into other spaces.
To recast the results in more familiar terms, we introduce some
notation. Suppose $X$ and $Y$ are normed spaces with bases $(x_i)$
and $(y_j)$ respectively. For any $v : X \to Y$, there exists a unique
family $(\alpha_{ij})$ so that $v x_i = \sum_j \alpha_{ij} y_j$.
We say that $v$ is \emph{real} (resp.~\emph{positive}) with respect to
these bases if, for any $i$ and $j$, we have $\alpha_{ij} \in \R$
(resp.~$\alpha_{ij} \geq 0$).
We often refer to the ``canonical'' basis of $\ell_1^N(\ell_\infty^n)$ consists of elements
$e_i \otimes f_j$, while that of $M_d$ -- of matrix units.

\begin{prop}\label{p:buhrman}
Suppose $n$ is a power of $2$, and let $N = \frac{2^n}{n}$. Then there exists an operator
$V : OH_{Nn} \to \ell_1^N(\ell_\infty^n)$, and a completely positive complete contraction
$U : \ell_1^N(\ell_\infty^n) \to M_n$, so that 
\begin{align*}
\|UV\|_{cb} =\Omega\Big(\frac{\sqrt{N}}{\log n}\Big)\text{     }\text{  and   }
 \text{     } \|V\| =O\Big(\sqrt{\frac{N}{n}}\Big).
\end{align*}
Moreover, $V$ is positive with respect to the canonical bases of its domain and
range, and $U$ is real with respect to the canonical bases.
\end{prop}

Denote by $(\delta_x)_{x\in \{0,1\}^n}$ the canonical orthonormal basis of 
$\ell_2^{Nn} = \ell_2(\{0,1\}^n)$.
Let
$$
M = \sum_x V \delta_x \otimes V \delta_x \in
 \ell_1^N(\ell_\infty^n) \otimes \ell_1^N(\ell_\infty^n) .
$$
By Proposition \ref{OH prop} and Corollary \ref{OH real}, \eqref{eq:good norm}, and \eqref{eq:good norm eps},
$\|M\|_{\eps} = \|V\|^2$, and $\|M\|_{min} = \|V\|_{cb}^2$.
The interested reader can check that the tensor $M$ has the form $N G_{KV}$, where $G_{KV}$
denotes the Khot-Vishnoi game (see \cite[Section 4]{BRSW} for a precise description;
in particular, $M$ has positive entries). 
\begin{proof}[Proof of Proposition \ref{p:buhrman}]
Consider the Cantor group $G = \{0,1\}^n$, and let $G_0$ be its Hadamard
subgroup, of cardinality $n$. Let $\Omega = G/G_0$ (then $|\Omega| = N$). It is important to note
(for future reference) that, for any distinct $x, y \in G_0$, $|x \ominus y| = 2^{n-1}$
(here $| \cdot |$ stands for the Hamming metric). Consider the probability measures
$\displaystyle \mu_0 = \frac{1+\svr}{2} \delta_0 + \frac{1-\svr}{2} \delta_1$ and
$\mu = \mu_0^{\otimes n}$, on $\{0,1\}$ and $G$ respectively (the number $\vr \in (0,1)$
will be specified later). The operator $V : \C^G \to \C^G$ is defined as the convolution
with the measure $\mu$: $Vf = C_\mu f = f * \mu$, or in other words,
$$
[Vf](x) = \sum_{y \in G}
 \left( \frac{1+\svr}{2} \right)^{n-|x \ominus y|} \left( \frac{1-\svr}{2} \right)^{|x \ominus y|} f(y)
$$
(all entries of the matrix representing $V$ are positive).
It is easy to see that Walsh functions are eigenvectors of $V$: $Vw_A = \vr^{|A|/2} w_A$.
For other properties of this operator, see e.g. \cite{W08}.
We view $V$ as acting from $OH_{Nn}$ to $\ell_1^N(\ell_\infty^n)$. The identification
of $OH_{Nn}$ with $\ell_2(G)$ is straightforward.
For the identification with $\ell_1^N(\ell_\infty^n)$,
given $g \in G$ we identify $\delta_g$ with $\delta_{[g]} \otimes \Phi_{[g]}(g)$, where for every $[g]$, $\Phi_{[g]}:[g]\rightarrow [n]$ defines a fixed enumeration of the elements in the class $[g]$.

Using the techniques of \cite{BRSW}, we prove that
$\|V\| \prec \sqrt{N}/\sqrt{n}$ and $\|UV\|_{cb} \succ \sqrt{N}/\log n$, where here we use symbols $ \prec$ and $\succ$ to denote inequality up to universal constants independent of the dimension. Indeed, consider the factorization $V = i_p C_\mu j_2$, where $j_2$ is the 
formal identity from $\ell_2^{Nn}$ to $L_2^{Nn} = L_2(G)$, where this last space is equipped with the
uniform probability measure on $G$, $C_\mu : L_2(G) \to L_p(G)$ is the convolution with $\mu$, $p = (1+\vr)/\vr$, and
$i_p : L_p(G) \to \ell_1^N(\ell_\infty^n)$ is the formal identity.
We clearly have $\|j_2\| = (Nn)^{-1/2}$. Furthermore, $\|i_p\| = N n^{1/p}$.
Indeed, this follows by noting that $$\|i_p\|\leq \|i_p: L_p(G)\rightarrow \ell_p^N(\ell_p^n)\|\|id:\ell_p^N(\ell_p^n)\rightarrow \ell_1^N(\ell_\infty^n)\|\leq (Nn)^{1/p}N^{1-1/p}=N n^{1/p}.$$Finally, $p$ is selected to make $C_\mu$ contractive, by Bonami-Beckner
Hypercontractivity Inequality (see e.g. \cite[Theorem 4.1]{W08}). This gives
$$
\|V\| \leq \|i_p\| \|C_\mu\| \|j_2\| = N n^{1/p} (Nn)^{-1/2} = N^{1/2} n^{1/p-1/2} =
 N^{1/2} n^{\vr/(1+\vr) - 1/2} .
$$

To define the operator $U$, consider, for each $x = (x_1, \ldots, x_n) \in G = \{0,1\}^n$,
the unit vector $$\displaystyle h_x =
\frac{1}{\sqrt{n}} \sum_{i=1}^n (-1)^{x_i} e_i \in \ell_2^n.$$
Here, $e_1, \ldots, e_n$ are the elements of the canonical basis in $\ell_2^n$.
Let $p_x = h_x \otimes h_x \in M_n$ be the orthogonal projection onto $\C h_x$, and
define $U : \ell_1^N(\ell_\infty^n)\to M_n : \delta_x \mapsto p_x$.
As noted in \cite{BRSW}, if $x, y \in G$ are distinct and belong to the same coset of
$\Omega = G/G_0$, then $|x \ominus y| = n/2$, hence $\langle h_x, h_y \rangle = 0$,
and consequently, $p_x p_y = 0$. Therefore, the restriction of $U$ to any copy of $\ell_\infty^n$
in its domain is a complete isometry. Thus, $U$ is a complete contraction.
Moreover, $U \delta_x$ has real entries for any $x$, hence $U$ is represented by a real matrix
(relative to canonical bases).

By construction, $U$ is real with respect to the standard bases.
Combining Proposition \ref{OH prop} with \eqref{eq:good norm}, we obtain $\|UV\|_{cb}^2 =
 \|\sum_{x \in G} UV \delta_x \otimes UV \delta_x\|_{M_n \otimes_{\min} M_n}$.
Identify $M_n \otimes_{\min} M_n$ with $M_{n^2}$, and consider the maximally entangled
state, defined as $\displaystyle f(A) = \langle A \xi, \xi \rangle$ for every $A\in M_{n^2}$, where
$$\displaystyle \xi = \frac{1}{\sqrt{n}} \sum_{i=1}^n e_i \otimes e_i.$$
For $a, b \in M_n$, $\displaystyle f(a \otimes b) =\frac{1}{n} \sum_{i,j=1}^n \langle e_j, a e_i \rangle \langle e_j, b e_i \rangle= \frac{1}{n} \trace(ab^{tr})$. Thus, $$\displaystyle f(p_x \otimes p_y) = \frac1n \langle h_x, h_y \rangle^2 =
 \frac{n - 2|x \ominus y|}{n^2}.$$

As $V = V^*$, $\displaystyle \sum_{x \in G} V \delta_x \otimes V \delta_x$ can be
identified on the vector space level with the operator $C_\mu^2 : \C^G \to \C^G$.
For any Walsh function $w_A$, we have $C_\mu^2 w_A = \vr^{|A|} w_A$, hence $C_\mu^2 = C_\nu$,
for the measure
$\displaystyle \nu = \left( \frac{1+\vr}{2} \delta_0 + \frac{1-\vr}{2} \delta_1 \right)^{\otimes n}$.
Therefore,
$$\sum_{x \in G} V \delta_x \otimes V \delta_x =
 \sum_{y,z \in G}  \left( \frac{1+\vr}{2} \right)^{n-|y \ominus z|}
 \left( \frac{1-\vr}{2} \right)^{|y \ominus z|} \delta_y \otimes \delta_z , $$
yielding
$$\sum_{x \in G} UV \delta_x \otimes UV \delta_x =
 \sum_{y,z \in G}  \left( \frac{1+\vr}{2} \right)^{n-|y \ominus z|}
 \left( \frac{1-\vr}{2} \right)^{|y \ominus z|} p_y \otimes p_z.$$
Consequently,
$$
\|UV\|_{cb}^2 = \left\| \sum_{x \in G} UV \delta_x \otimes UV \delta_x \right\| \geq
 \sum_{y,z \in G}  \left( \frac{1+\vr}{2} \right)^{n-|y \ominus z|}
 \left( \frac{1-\vr}{2} \right)^{|y \ominus z|} f(p_y \otimes p_z)
$$
$$
= \frac1n \sum_{y,z \in G}  \left( \frac{1+\vr}{2} \right)^{n-|y \ominus z|}
 \left( \frac{1-\vr}{2} \right)^{|y \ominus z|} \left( 1 - 2 \frac{|z \ominus y|}{n} \right)^2
$$
$$
= N \sum_{a \in G} \left( \frac{1+\vr}{2} \right)^{n-|a|}
 \left( \frac{1-\vr}{2} \right)^{|a|} \left( 1 - 2 \frac{|a|}{n} \right)^2 =
 N \expe \left( 1 - 2 \frac{X}{n} \right)^2 ,
$$
where $X$ is polynomially distributed with parameters $n$ and $\big(1+\vr\big)/2$; that is,
$$
\prob \Big( X = k \Big) = \left( \frac{1+\vr}{2} \right)^{n-k}
 \left( \frac{1-\vr}{2} \right)^{k} \binom{n}{k} .
$$
It is well known that the expected value and the variance of $X$ are given by
$\displaystyle n \frac{1-\vr}{2}$ and $\displaystyle n \frac{1-\vr}{2} \frac{1+\vr}{2}$
respectively. Hence, 
\begin{align*}
\expe \left( 1 - 2 \frac{X}{n} \right)^2 &= 1 - \frac2n \expe(X) + \frac{1}{n^2} \expe(X^2) \\&=
1 - 4 \frac{1-\vr}{2} + \frac{4}{n^2} \left[ \left( n \frac{1-\vr}{2} \right)^2 +
 n \frac{1-\vr}{2} \frac{1+\vr}{2} \right]\\& >
\left( 1 - 2 \frac{1-\vr}{2} \right)^2 = \vr^2 .
\end{align*}
Thus, $\|UV\|_{cb} \geq N^{1/2} \vr$.

Meanwhile, $\displaystyle \|V\| \leq  N^{1/2} n^{\vr - 1/2}$. Set $\vr \sim 1/\log n$, we obtain:
$\|UV\|_{cb} \succ N^{1/2}/\log n$, $\displaystyle \|V\| \prec  N^{1/2}/ n^{1/2}$.
\end{proof}
The main result of this work follows from the following result. 
\begin{prop}\label{p:from OH}
Suppose $n$ is a positive power of $2$. Then there exist $m \leq c n^8$, $s \leq n^2$,
($c$ is an absolute constant), an operator $T : OH_s \to \ell_1^m(\ell_\infty^n)$
and a complete contraction $S : \ell_1^m(\ell_\infty^n) \to M_n$, so that
$$\frac{\|ST\|_{cb}}{\|T\|}=\Omega\Big(\frac{\sqrt{n}}{\log n}\Big).$$

Moreover,  one can select an orthonormal basis $(\theta_i)_{i=1}^s$ in $OH_s$
so that, for $1 \leq i \leq s$, $T\theta_i$ and $ST\theta_i$ have real coefficients
relative to the canonical bases of the spaces $\ell_1^m(\ell_\infty^n)$ and $M_n$
respectively.
\end{prop}

Now set $$M:=\frac{1}{\|T\|^2}\sum_{i=1}^{s} T\theta_i \otimes T\theta_i
 \in \ell_1^m(\ell_\infty^n) \otimes \ell_1^m(\ell_\infty^n).$$
According to Proposition \ref{OH prop}, Corollary \ref{OH real},
\eqref{eq:good norm}, and \eqref{eq:good norm eps},
we have
$$\|M\|_{\epsilon}=1\text{     }\text{  and   }\text{     }  \|M\|_{\min}\geq
\|(S \otimes S)M\|_{M_n \otimes_{\min} M_n} =\Omega\Big(\frac{n}{\log^2 n}\Big).$$
Our main Theorem \ref{main theorem} follows now from Propositions
\ref{connection min-viol} and \ref{p:Regev}.

\begin{proof}[Proof of Proposition \ref{p:from OH}]
We need to ``reduce the dimension'' from $N$ to $m$.
Let $H$ be the orthogonal complement of $\ker UV$ in $OH_{Nn}$.
Note that $s:= \dim H \leq \rank V \leq n^2$. Moreover, as $UV$ has real
coefficients (with respect to the canonical bases),
$H$ is the range of $(UV)^t$ (the transpose of $UV$).
Consequently, $H$ is spanned by real linear combinations of
the canonical basis of $OH_{Nn}$. A Gram-Schmidt procedure yeilds
an orthonormal basis in $H$ whose elements have real coefficients
relative to the canonical basis of $OH_{Nn}$. Therefore, there
exists an isometry $R : OH_s \to H$ with real coefficients (relative
to the canonical basis $(\theta_i)_{i=1}^s$ of $OH_s$, and the canonical
basis $(\delta_x)$ of $OH_{Nn}$). Let $\tilde{V} = VR$. Then
$\|\tilde{V}\| \leq \|V\| =O\big( \sqrt{N}/\sqrt{n}\big)$, and
(due to the homogeneity of $OH_{Nn}$)
$\|U \tilde{V}\|_{cb} = \|UV\|_{cb} =\Omega\big(\sqrt{N}/\log n\big)$.
Let $E = \rank \tilde{V}$, which is an $n^2$-dimensional subspace of $\ell_1^N(\ell_\infty^n)$. 

We will consider the space $S_1^n[E]$, which is known to be a subspace of
$S_1^n[\ell_1^N(\ell_\infty^n)]$. Moreover, this last space is completely isometric to
$\ell_1^N(S_1^n[\ell_\infty^n])$ via the natural identification. Hence, if we denote
$\tilde{E}=S_1^n[E]$ and $X=S_1^n[\ell_\infty^n]$, we have that the $n^4$ dimensional space
$\tilde{E}$ is a subspace of $\ell_1^N(X)$. We can then apply Proposition \ref{p:reduction}
to deduce that the map $J$ introduced in Remark \ref{remark: map J} verifies that
$J\otimes id_X:\ell_1^N(X)\rightarrow \ell_1^m(X)$ is a $\frac{1}{2}$-embedding when
it is restricted to the subspace $\tilde{E}$ and $m\simeq n^8$. That is, the map
$(J\otimes \ell_\infty^n)|_E\otimes id_{S_1^n}:S_1^n[E]\rightarrow \ell_1^N(S_1^n[\ell_\infty^n])$
defines a $\frac{1}{2}$-isomorphism. The following diagram gives us the picture:
$$
\xymatrix@R=1.25cm@C=2.25cm {{\ell_1^N(S_1^n[\ell_\infty^n])}\ar[r]^{J\otimes id_{S_1^n}\otimes id_{\ell_\infty^n}} & {\ell_1^m(S_1^n[\ell_\infty^n])} \\
{S_1^n[E]}\ar@{^{(}->}[u]\ar[r]^{(J\otimes \ell_\infty^n)|_E\otimes id_{S_1^n}} & {S_1^n[(J\otimes \ell_\infty^n)(E)].}\ar@{^{(}->}[u]}
$$
Let us define the linear map $T : OH^{n^2} \to \ell_1^m(\ell_\infty^n)$ given by
$T=(J\otimes id_{\ell_\infty^n}) \circ \tilde{V}$. Since $E=\rank \tilde{V}$ and
$(J\otimes \ell_\infty^n)|_E$ must be an $\frac{1}{2}$-embedding, we immediately obtain
that $\|T\|\lesssim \|\tilde{V}\|=O\big(\sqrt{N}/\sqrt{n}\big)$. On the other hand,
if we denote $F=(J\otimes id_{\ell_\infty^n})(E)$ we can define the map
$S=U\circ (J^{-1}\otimes id_{\ell_\infty^n}): F\rightarrow M_n$. Now, since $U$
is completely contractive and the map
$id_{S_1^n}\otimes (J^{-1}\otimes id_{\ell_\infty^n})):
 S_1^n[F]\rightarrow S_1^n[\ell_1^N(\ell_\infty^n)]$
is a $\frac{1}{2}$-embedding, we deduce that $\|S\|_{cb}=O(1)$. Indeed, to conclude
this we have used that for every linear map $S$ into $M_n$ we have
$\|S\|_{cb}=\|id_{S_1^n}\otimes S\|$ (see \cite[Theorem 1.5+Lemma 1.7]{PiLp}).
Moreover, we can extend $S$ to an operator from $\ell_1^m(\ell_\infty^n)$ to $M_n$
(also denoted by $S$) without increasing its cb norm (see \cite[Cor. 1.7]{PiOS}).

Finally, note that
\begin{align*}
\|ST\|_{cb}=
\|U\circ (J^{-1}\otimes id_{\ell_\infty^n)}\circ (J\otimes id_{\ell_\infty^n}) \circ \tilde{V}\|_{cb}=
\|U \tilde{V}\|_{cb}=\|U V\|_{cb}=\Omega\Big(\frac{\sqrt{N}}{\log n}\Big).
\end{align*}
As $J$ has real (in fact, positive) coefficients, we are done.
\end{proof}
\subsection{Some comments on our construction}\label{sec comments}
Let us finish this section with some final comments about our results. First, we notice that our procedure starts with the Khot-Vishnoi game $G_{KV}$, which has a particularly nice structure in $\ell_1^{N}(\ell_\infty^{n})\otimes \ell_1^{N}(\ell_\infty^{n})$ (it has for instance positive coefficients and it is completely explicit), and outputs another element $M	\in \ell_1^{m}(\ell_\infty^{n})\otimes \ell_1^{m}(\ell_\infty^{n})$ with a more obscure description. There are several reasons for this. First of all, in order to apply the empirical method in the form of Proposition \ref{p:reduction}, we need to start with a low rank element, and the Khot-Vishnoi game has a very large rank. In order to save this obstacle, we must cut the Khot-Vishnoi game, which is done by considering the element $\tilde{V} = VP$ in the proof of Proposition \ref{p:from OH}. Unfortunately, composing with the general projection $P$ involves a lack of control on the structure of the new object. An important point is that we cannot 
assure that the new element $M$ (defined via $\tilde{V}$) has positive coefficients  when seen as an element in $\ell_1^{m}(\ell_\infty^{n})\otimes \ell_1^{m}(\ell_\infty^{n})$ (even though our starting point $G_{KV}$ did). This limitation is the main reason to consider the bias of the game (which allows us to work with nonpositive elements via the correspondence with Bell inequalities explained in the introduction). We think that proving a reduction atoms method which preserves positivity is a very interesting problem.

Another disadvantage of our result is that it introduces some randomness. Indeed, while the Khot-Vishnoi game is a completely explicit element in $\ell_1^{N}(\ell_\infty^{n})\otimes \ell_1^{N}(\ell_\infty^{n})$, our final element in $\ell_1^{m}(\ell_\infty^{n})\otimes \ell_1^{m}(\ell_\infty^{n})$ is not explicit since the map $J$ explained in Remark \ref{remark: map J} has a probabilistic nature. 

Finally, in order to obtain Theorem \ref{main theorem} we must use Proposition \ref{connection min-viol}, which implies a lack of knowledge about the quantum probability distribution to be used in order to lower bound the vale $\omega^*(M)$. However, a careful study of the proof of Proposition \ref{p:from OH} allows to see that one can use the $n$-dimensional maximally entangled in Theorem \ref{main theorem}.

On the other hand, the procedure used in this paper is very general and it can be applied in many different contexts. There are two key points in our proof. First of all, the fact that the Khot-Vishnoi game is an \emph{OH-game}, in the sense that it can be seen as an element of the form $VV^*$, where $V : OH_{Nn} \to \ell_1^N(\ell_\infty^n)$. The second crucial element in our proof is that, although the game $G_{KV}$ has a very large rank, the dimension used in the corresponding quantum strategy to lower bound the value $\omega^*(G_{KV})$ is of order $n$. The previous two ingredients allow to ``cut the Khot-Vishnoi game'' so that we obtain a new element with lower rank and essentially the same classical and quantum values. Although factorizing through a Hilbert space $OH_N$ can be understood as a very restrictive property, most of the games used to obtain large Bell violations have this characteristic (as those in \cite{JPI, JPPVW}). Hence, this property seems to be very natural when studying extreme objects 
and it is very plausible that our method is of independent interest in some other contexts such that the study of integrality gaps between the classical and quantum value of a game and certain SDP-relaxations.

Finally, we note that exactly the same approach followed in this work can be applied to the recent paper \cite{PaYi} to prove that Bell violations of order $\sqrt{n}/\log^2n$ can be obtained by only using binary questions in one party and with the additional property that only a polynomial number of questions are needed. 
\section{A potential improvement of Proposition \ref{p:reduction}}\label{sec: improving empirical method}

Can we do better than in Section \ref{s:empirical}? Yes, if we follow \cite[Section 3]{BLM},
and control the type of our subspace of $L_1(X)$.

First re-state \cite[Lemma 3.3]{BLM}.
\begin{prop}\label{p:type L1}
Suppose $X$ is a Banach space, $\vr \in (0,1/2)$, and $E$ is an $m$-dimensional subspace of
$\ell_1^N(X)$, with $N \geq m/\vr^2$. If
$$
K \geq c T_p(X) \vr^{-2} \frac{\big( \log \vr^{-1} \big)^{1-1/p}}{p-1}
 m \left( \frac{N}{m} \right)^{1/p} \left( \log \frac{N}{m} \right)^{1/p} ,
$$
then $\ell_1^K(X)$ contains a space $E^\prime$ which is $(1+\vr)$-isomorphic to $E$.
\end{prop}
\begin{remark}\label{r:map2}
As in Section \ref{s:empirical}, the $(1+\vr)$-isomorphism
from $E$ to $E^\prime$ we are constructing is actually the restriction to $E$ of the
truncation/change of density map $\ell_1^N(X) \to \ell_1^K(X) :
 f \mapsto \big(\alpha_i f_{s_i}\big)_{i=1}^K$.
\end{remark}

Now suppose $E$ is an $m$-dimensional subspace of $L_1(X)$. For $\vr > 0$, define
$N(E,\vr)$ to be the smallest $N$ so that $\ell_1^N(X)$ contains a $(1+\vr)$-isomorphic
copy of $E$. The following is a vector-valued version of \cite[Lemma 3.3']{BLM}
(and immediately follows from Proposition \ref{p:type L1}):

\begin{prop}\label{p:type N}
If $N(E,\delta) \geq 4 m \vr^{-2}$, then
$$
N(E,\vr+\delta) \geq c T_p(X) \vr^{-2} \frac{\big( \log \vr^{-1} \big)^{1-1/p}}{p-1}
 m \left( \frac{N(E,\delta)}{m} \right)^{1/p} \left( \log \frac{N(E,\delta)}{m} \right)^{1/p} .
$$
\end{prop}

Applying the empirical method of Section \ref{s:empirical} and iterating, we obtain:

\begin{prop}\label{p:type end}
For any $\rho > 2$, any $m$-dimensional $E \subset L_1(X)$, and any $\vr > 0$, we have
$\displaystyle N(E,\vr) \leq C(p, T_p(E), \rho) \vr^{\rho p/(p-1)} m$.
\end{prop}

\begin{remark}\label{r:map3}
As before, we pass from $E$ to its $(1+\vr)$-isomorphic copy in $\ell_1^{N(E,\vr)}(X)$
using a change of density.
\end{remark}

The proof of Proposition \ref{p:type L1} very closely follows \cite[Section 3]{BLM}
(which is itself a variation of the
``empirical method''). The only missing ingredient is a vector-valued version of
Pisier Factorization Theorem, which can be obtained quite easily.
To formulate the theorem, we introduce some notation. If $\mu$ is a measure on $\Omega$, then 
$D(\mu) = \{f \in L_1(\mu) : \int f \, d\mu = 1, f \geq 0\}$ is the set of densities.
We continue using the convention $0/0 = 0$.

\begin{theorem}\label{t:pisier}
Suppose $X$ is a Banach space, $0 < r < p < \infty$, and $(x_i)_{i \in I}$ is a
subset of $L_r(\mu,X)$. Then the following are equivalent:
\begin{enumerate}
\item 
There exists a constant $C_1$ and $f \in D(\mu)$ so that, for any $\mu$-measurable set $E$,
and for any $i \in I$, $\|1_E x_i\|_{L_r(\mu,X)} \leq C_1 \Big( \int_E f \, d\mu \Big)^{1/r - 1/p}$.
\item
There exists a constant $C_2$ and $f \in D(\mu)$ so that, for any $i \in I$,
$\|f^{-1/r} x_i\|_{L_{p\infty}(\mu,X)} \leq C_2$.
\item
There exists a constant $C_3$ so that, for any finite sequence $(\alpha_i)$,
\\
$\|\sup_i \{|\alpha_i| \|x_i\|_X\} \|_{L_p(\mu)} \leq C_3 \Big( \sum_i |\alpha_i|^p \Big)^{1/p}$.
\end{enumerate}
Moreover, the constants are proportional to each other. For instance, if (1) holds,
then (2) and (3) also hold, with $C_2, C_3 \leq c C_1$, where $c$ is a universal constant.
\end{theorem}

In part (3), $\sup_i \{|\alpha_i| \|x_i\|_X\}$ is a scalar-valued function defined
almost everywhere on $\Omega$ via
$[\sup_i \{|\alpha_i| \|x_i\|_X\}](\omega) = \sup_i \{ |\alpha_i| \|x_i(\omega)\|_X\}$. The proof is obtained by applying \cite[Theorem 1.1]{Pi86} to the scalar-valued functions
$y_i(\cdot) = \|x_i(\cdot)\|_X \in L_r(\mu)$.

Next state a vector-valued counterpart of \cite[Theorem 1.2]{Pi86}.

\begin{corollary}\label{c:op pisier}
Suppose $X, Z$ are Banach space, $0 < r < p < \infty$, and $T \in B(Z,L_r(\mu,X))$.
Then the following are equivalent:
\begin{enumerate}
\item
There exists a constant $C_1$ and $f \in D(\mu)$ so that, for any $\mu$-measurable set $E$,
and for any $z \in Z$, $\|1_E Tz\|_{L_r(\mu,X)} \leq C_1 \|z\| \Big( \int_E f \, d\mu \Big)^{1/r - 1/p}$.
\item
There exists a constant $C_2$ and $f \in D(\mu)$ so that, for any $z \in Z$,
\\ $\|f^{-1/r} T z\|_{L_{p\infty}(\mu,X)} \leq C_2 \|z\|$.
\item 
There exists a constant $C_3$ so that, for any finite sequence $(z_i)$,
\\
$\|\sup_i \|Tz_i\|_X\|_{L_r(\mu)} \leq C_3 \big( \sum_i \|z_i\|^p \big)^{1/p}$.
\item
The operator $T$ admits a factorization $T = M_f \circ T^\prime$, where $f \in D(\mu)$,
$T^\prime : Z \to L_{p\infty}(f \cdot \mu, X)$ has norm not exceeding $C_3$, and
$M_f : L_{p\infty}(f \cdot \mu, X) \to L_r(\mu,X)$ is the operator of multiplication by $f$.
It is easy to see that $\|M_f\|$ is bounded above by a universal constant.
\end{enumerate}
Moreover, the constants are proportional to each other.
\end{corollary}

To prove this statement, simply apply Theorem \ref{t:pisier} to the image of the unit ball of $Z$.

We turn then to a vector-valued version of \cite[Remark 1.4]{Pi86}.

\begin{corollary}\label{c:type}
Suppose $Z$ is a Banach space of type $p > 1$. Then for any $T \in B(Z,L_1(\Omega,\mu,X))$
there exists a density $f \in D(\mu)$ so that
$$
\{\omega \in \Omega : f(\omega) = 0\} \subset \cap_{z \in Z}
\{\omega \in \Omega : [Tz](\omega) = 0 \} ,
$$
and, for every $z \in Z$, $\|f^{-1} Tz\|_{L_{p\infty}(\Omega,f\mu,X)} \leq e \|T\| \|z\|$.
\end{corollary}

We would like to apply this when $Z$ is a subspace of $L_1(\Omega,\mu,X)$, and
$T$ is the identity map. But, we do not know how!
\begin{proof}[Sketch of a proof]

First assume $X$ is the field of scalars. By \cite[Proposition 11.10]{DJT}, $\pi_{q1}(I_{Z^*}) \leq C_q(Z^*) \leq T_p(Z)$
(here $1/p + 1/q = 1$). Recall that $\pi_{q1}(I_{Z^*})$ is the smallest $C$ with
the property that, for any $z_1^*, \ldots, z_n^*$,
$$\big(\sum_{i=1}^n \|z_i^*\|^q\big)^{1/q} \leq C \max \Big\{\big\|\sum_{i=1}^n w_i z_i^*\big\|: |w_i|=1, i=1,\cdots, n \Big\}.$$

Set $C = \pi_{q1}(I_{Z^*})$, and consider a contraction $T : Z \to L_1(\mu)$.
We need to show that case (3) of Corollary \ref{c:op pisier} holds, with $C_3 = C$.
In other words, we have to show that, for any $n$,
$T_n : \ell_p^n(Z) \to \widetilde{L_1(\ell_\infty^n)}$ has norm not exceeding $C$.
Here, we need to recall some notation from e.g.~\cite{LT2}: if $E$ is a Banach lattice,
and $1 \leq s \leq \infty$, then
$$
\left\|(f_1, \ldots, f_n)\right\|_{\widetilde{E(\ell_s^n)}} =
\left\| \Big( \sum_{i=1}^n |f_i|^s \Big)^{1/s} \right\|_E
$$
(with the obvious modification in the case $s = \infty$).
By duality, it suffices to show the estimate
$\|T_n^* : \widetilde{L_\infty(\ell_1^n)} \to \ell_q^n(Z^*)\| \leq C$
(here $1/p + 1/q = 1$). In other words, we have to show that, if the functions
$f_1, \ldots, f_n \in L_\infty(\mu)$ satisfy $\|\sum_{i=1}^n |f_i| \|_\infty \leq 1$,
then $\sum_{i=1}^n \|T^* f_i\|^p \leq C^p$. But
\begin{align*}
\big\|\sum_{i=1}^n |f_i| \big\|_\infty &=
\Big\|\max \Big\{\Big|\sum_{i=1}^n w_i f_i\Big |: |w_i|=1, i=1,\cdots, n \Big\}\Big\|_\infty\\
&\geq \max \Big\{\Big\|\sum_{i=1}^n w_i f_i\Big\|_\infty: |w_i|=1, i=1,\cdots, n \Big\}.
\end{align*}

As $T$ is a contraction,
$$\max \Big\{\Big\|\sum_{i=1}^n w_i Tf_i\Big\| :|w_i|=1, i=1,\cdots, n \Big\}\leq \max \Big\{\Big\|\sum_{i=1}^n w_i f_i\Big\|_\infty:|w_i|=1, i=1,\cdots, n \Big\} \leq 1.$$
The definition of $C = \pi_{q1}(I_{Z^*})$ finishes the proof of Corollary \ref{c:op pisier}(3)
in the case when $X$ is the field of scalars.

Now consider the case of general $X$. Consider $\phi \in L_\infty(\mu,X^*)$ with $\|\phi\| \leq 1$, and the
operator $T_\phi : Z \to L_1(\mu) : [Tz](\cdot) = \langle \phi(\cdot), [Tz](\cdot) \rangle$.
Clearly, $\|T_\phi\| \leq \|T\| \leq 1$. By the previous paragraph,
$\big\| \sup_i |T_\phi z_i| \big\|_{L_1(\mu)} \leq C \big( \sum_{i=1}^n \|z_i\|^p \big)^{1/p}$
It remains to note that $\phi$ can be selected to make $\big\| \sup_i |T_\phi z_i| \big\|_{L_1(\mu)}$
arbitrarily close to $\big\|\sup_i \|Tz_i\|_X \big\|_{L_1(\mu)}$.
\end{proof}

\hfill \noindent \textbf{Marius Junge} \\
\null \hfill Department of Mathematics \\ \null \hfill
University of Illinois at Urbana-Champaign\\ \null \hfill 1409 W. Green St. Urbana, IL 61891. USA
\\ \null \hfill\texttt{junge@math.uiuc.edu}

\vskip 0.5cm

\hfill \noindent \textbf{Timur Oikhberg } \\
\null \hfill Department of Mathematics \\ \null \hfill
University of Illinois at Urbana-Champaign\\ \null \hfill 1409 W. Green St. Urbana, IL 61891. USA
\\ \null \hfill\texttt{oikhberg@illinois.edu}

\vskip 0.5cm

\hfill \noindent \textbf{Carlos Palazuelos} \\
\null \hfill Instituto de Ciencias Matem\'aticas, ICMAT\\
\null \hfill Facultad de Ciencias Matem\'aticas\\ \null \hfill
Universidad Complutense de Madrid \\ \null \hfill Plaza de Ciencias s/n.
28040, Madrid. Spain
\\ \null \hfill\texttt{carlospalazuelos@mat.ucm.es}

\end{document}